\documentclass[a4paper,conference]{IEEEtran}

\IEEEsettopmargin{t}{30mm}
\IEEEquantizetextheight{c}
\IEEEsettextwidth{14mm}{14mm}
\IEEEsetsidemargin{c}{0mm}

\usepackage[utf8]{inputenc}
\usepackage[T1]{fontenc}
\usepackage{url}
\usepackage{ifthen}
\usepackage{cite}
\usepackage[cmex10]{amsmath} 
\usepackage{amsmath,amssymb,amsfonts}
\usepackage{tabularx,booktabs}
\usepackage{mathtools}
\usepackage{algorithmic}
\usepackage{graphicx}
\usepackage{textcomp}
\usepackage{xcolor}
\def\BibTeX{{\rm B\kern-.05em{\sc i\kern-.025em b}\kern-.08em
		T\kern-.1667em\lower.7ex\hbox{E}\kern-.125emX}}
\newcolumntype{C}{>{\centering\arraybackslash}X} 
\newcolumntype{b}{>{\hsize=2.3\hsize}X}
\usepackage{amsthm}
\usepackage{bbm}
\usepackage{csquotes}
\usepackage{chngcntr}
\usepackage{apptools}

\usepackage{cite}

\theoremstyle{plain}
\newtheorem{theorem}{Theorem}
\newtheorem{lemma}{Lemma}
\newtheorem{corollary}{Corollary}

\theoremstyle{definition}
\newtheorem{dfn}{Definition}

\theoremstyle{remark}

\newtheorem{remark}{Remark}

\newcommand{\A}{\mathcal{A}}

\newcommand{\X}{\mathcal{X}}
\newcommand{\Y}{\mathcal{Y}}
\newcommand{\Z}{\mathcal{Z}}

\newcommand{\E}{\mathbb{E}}

\newcommand{\Pm}{\mathcal{P}}
\newcommand{\Q}{\mathcal{Q}}
\newcommand{\F}{\mathcal{F}}

\DeclareMathOperator*{\esssup}{ess\,sup}
\newcommand{\ml}[2]{\mathcal{L}\left(#1  \!\!  \to  \!\!   #2\right)} 
\begin{document}
\title{Robust Generalization via $\alpha-$Mutual Information
}

\author{
	\IEEEauthorblockN{Amedeo Roberto Esposito, Michael Gastpar}
	\IEEEauthorblockA{\textit{School of Computer and Communication Sciences} \\
		EPFL\\
		\{amedeo.esposito, michael.gastpar\}@epfl.ch}
	
	\and
	\IEEEauthorblockN{Ibrahim Issa}
	\IEEEauthorblockA{\textit{Electrical and Computer Engineering Department} \\
		American University of Beirut\\
		ii19@aub.edu.lb}
}

\maketitle

\begin{abstract}
	The aim of this work is to provide bounds connecting two probability measures of the same event using R\'enyi $\alpha$-Divergences and Sibson's $\alpha$-Mutual Information, a generalization of respectively the Kullback-Leibler Divergence and Shannon's Mutual Information. A particular case of interest can be found when the two probability measures considered are a joint distribution and the corresponding product of marginals (representing the statistically independent scenario). In this case a bound using Sibson's $\alpha-$Mutual Information is retrieved, extending a result involving Maximal Leakage to general alphabets. These results have broad applications, from bounding the generalization error of learning algorithms to the more general framework of adaptive data analysis, provided that the divergences and/or information measures used are amenable to such an analysis ({\it i.e.,} are robust to post-processing and compose adaptively). The generalization error bounds are derived with respect to high-probability events but a corresponding bound on expected generalization error is also retrieved. 
\end{abstract}


\section{Introduction}
Let us consider two probability spaces $(\Omega,\F,\Pm),(\Omega,\F,\Q)$ and let $E\in \F$ be a measurable event. Given some divergence between the two distributions $\hat{D}(\Pm,\Q)$  (e.g., KL, R\'enyi's $\alpha-$Divergence, ...) our aim is to provide bounds of the following shape: 
\begin{equation}\Pm(E) \leq f(\Q(E))\cdot g(\hat{D}(\Pm,\Q)),\label{generalBound}\end{equation} for some functions $f,g$. $E$ represents some \enquote{undesirable} event 
(e.g., large generalization error), whose measure under $\Q$ is known and whose measure under $\Pm$ we wish to bound. To that end, we use some notion of ``distance'' between $\Pm$ and $\Q$.  
Of particular interest is the case where $\Omega = \X\times\Y$,  $\Pm = \Pm_{XY}$ (the joint distribution), and $\Q=\Pm_X\Pm_Y$ (product of the marginals). This allows us to bound the likelihood of $E \subseteq \X \times \Y$ when two random variables $X$ and $Y$ are dependent as a function of the likelihood of $E$ when $X$ and $Y$ are independent (typically easier to analyze). 
Indeed, an immediate application can be found in bounding the generalization error of a learning algorithm and, when the proper measure is chosen, in adaptive data analysis. In order to be used in adaptive data analysis, such measure needs to be robust to post-processing and to compose adaptively (meaning that we can bound the measure between input and output of the composition of a sequence of algorithms if each of them has bounded measure). Results of this form involving mutual information can be found in \cite{learningMI,infoThGenAn,explBiasMI}. More recently, a different measure satisfying these properties, maximal leakage \cite{leakage},  has been used in \cite{ISIT2019,ITW2019}. More specifically, it was shown that Equation~\eqref{generalBound} holds
for the following choice of $f(\Pm_X\Pm_Y(E))=\max_y(\Pm_X(E_y))$ and $g(\hat{D}(\Pm_{XY}||\Pm_X\Pm_Y)) = \exp(\ml{X}{Y}) = \mathbb{E}_Y\left(D_\infty(\Pm_{X|Y}||\Pm_X)\right)= I_\infty(X;Y)$, where $I_\infty(X;Y)$ is the Sibson mutual information of order infinity. 
In this work, we derive a general bound in the form of~\eqref{generalBound} and focus on two interesting special cases. In particular, one specialization of the bound leads to a family of bounds in terms of $\alpha$-divergences. The other specialization leads to a family of bounds in terms of Sibson's $\alpha$-mutual information, thus generalizing the previous maximal leakage bound (which corresponds to $\alpha \to \infty$).
\section{Background And Definitions}

\subsection{Sibson's $\alpha-$Mutual Information}
Introduced by R\'enyi as a generalization of entropy and KL-divergence, $\alpha$-divergence has found many applications ranging from hypothesis testing to guessing and several other statistical inference problems~\cite{verduAlpha}. Indeed, it has several useful operational interpretations (e.g., the number of bits by which a mixture of two codes can be compressed, the cut-off rate in block coding and hypothesis testing \cite{alphaDiv,opMeanRDiv1}\cite[p. 649]{opMeanRDiv2}). It can be defined as follows~\cite{alphaDiv}.
\begin{dfn}
	Let $(\Omega,\F,\Pm),(\Omega,\F,\Q)$ be two probability spaces. Let $\alpha>0$ be a positive real number different from $1$. Consider a measure $\mu$ such that $\Pm\ll\mu$ and $\Q\ll\mu$ (such a measure always exists, e.g. $\mu=(\Pm+\Q)/2$)) and denote with $p,q$ the densities of $\Pm,\Q$ with respect to $\mu$. The $\alpha-$Divergence of $\Pm$ from $\Q$ is defined as follows:
	\begin{align}
	D_\alpha(\Pm\|\Q)=\frac{1}{\alpha-1} \ln \int p^\alpha q^{1-\alpha} d\mu.
	\end{align}
\end{dfn}
\begin{remark}
	The definition is independent of the chosen measure $\mu$. 
	It is indeed possible to show that
	$\int p^{\alpha}q^{1-\alpha} d\mu = \int \left(\frac{q}{p}\right)^{1-\alpha}d\Pm $, and that whenever $\Pm\ll\Q$ or $0<\alpha<1,$ we have $\int p^{\alpha}q^{1-\alpha} d\mu= \int \left(\frac{p}{q}\right)^{\alpha}d\Q$, see \cite{alphaDiv}.
\end{remark}

It can be shown that if $\alpha>1$ and $\Pm\not\ll\Q$ then $D_\alpha(\Pm\|\Q)=\infty$. The behaviour of the measure for $\alpha\in\{0,1,\infty\}$ can be defined by continuity. In general, one has that $D_1(\Pm\|\Q) = D(\Pm\|\Q)$ but if $D(\Pm\|\Q)=\infty$ or there exists $\beta$ such that $D_\beta(\Pm\|\Q)<\infty$ then $\lim_{\alpha\downarrow1}D_\alpha(\Pm\|\|Q)=D(\Pm\|\Q)$\cite[Theorem 5]{alphaDiv}. For an extensive treatment of $\alpha$-divergences and their properties we refer the reader to~\cite{alphaDiv}. 
Starting from the concept of $\alpha-$divergence, Sibson built a generalization of mutual information that retains many interesting properties. The definition is the following \cite{verduAlpha}:
\begin{dfn}
	Let $X$ and $Y$ be two random variables jointly distributed according to $\Pm_{XY}$, and with marginal distributions $\Pm_X$ and $\Pm_Y$, respectively. For $\alpha>0$, the Sibson's mutual information of order $\alpha$ between $X$ and $Y$ is defined as:
	\begin{align}
	I_\alpha(X;Y) = \min_{Q_Y} D_\alpha(\Pm_{XY}\|\Pm_X Q_Y) \label{iAlphaDef}.
	\end{align}
\end{dfn}
Moreover, $\lim_{\alpha\to 1}I_\alpha(X;Y)=I(X;Y)$. On the other hand when $\alpha\to\infty$, we get: $$I_\infty(X;Y)=\log\mathbb{E}_{\Pm_Y}\left[\sup_{x:\Pm_X(x)>0} \frac{\Pm_{XY}(\{x,Y\})}{\Pm_X(\{x\})\Pm_Y(\{Y\})}\right].$$
For more details on Sibson's $\alpha$-MI we refer the reader to~\cite{verduAlpha}.

\subsection{Learning Theory}
In this section, we provide some basic background knowledge on learning algorithms and concepts like generalization error. We are mainly interested in supervised learning, where the algorithm learns a \emph{classifier} by looking at points in a proper space and the corresponding labels. 

More formally, suppose we have an instance space $\mathcal{Z}$ and a hypothesis space $\mathcal{H}$. The hypothesis space is a set of functions that, given a data-point $s\in \mathcal{Z}$ give as an output the corresponding label $\mathcal{Y}$. Suppose we are given a training data set $\mathcal{Z}^n\ni S =\{z_1,\ldots,z_n\}$ made of $n$ points sampled in an i.i.d fashion from some distribution $\mathcal{P}$. Given some $n\in\mathbb{N}$, a learning algorithm is a (possibly stochastic) mapping $\mathcal{A}:\mathcal{Z}^n\to \mathcal{H}$ that given as an input a finite sequence of points $S\in\mathcal{Z}^n$ outputs some classifier $h=\mathcal{A}(S)\in\mathcal{H}$.
In the simplest setting we can think of $\mathcal{Z}$ as a product between the space of data-points and the space of labels, {\it i.e.,} $\mathcal{Z}=\mathcal{X}\times\mathcal{Y}$ and suppose that $\mathcal{A}$ is fed with $n$ data-label pairs $(x,y)\in\mathcal{Z}$. In this work we will view $\mathcal{A}$ as a family of conditional distributions $\mathcal{P}_{H|S}$ and provide a stochastic analysis of its generalization capabilities using the information measures introduced above.
The goal 
is to generate a hypothesis $h:\mathcal{X}\to \mathcal{Y}$ that has good performance on both the training set and newly sampled points from $\mathcal{X}$. In order to ensure such property the concept of generalization error is introduced.

\begin{dfn} Let $\mathcal{P}$ be some distribution over $\mathcal{Z}$. Let $\ell:\mathcal{H}\times\mathcal{Z}\to\mathbb{R}$ be a loss function. The error (or risk) of a prediction rule $h$ with respect to $\mathcal{P}$ is defined as \begin{equation}L _\mathcal{P}(h)=\mathbb{E}_{Z\sim \mathcal{P}}[\ell(h,Z)], \end{equation}
	while, given a sample $S=(z_1,\ldots,z_n)$, 
	the empirical error
	of $h$ with respect to $S$ is defined as \begin{equation}\label{genEmpRisk}L_{S}(h) = \frac1n \sum_{i=1}^n \ell(h, z_i).\end{equation}
	Moreover, given a learning algorithm $\mathcal{A}:\mathcal{Z}^n\to\mathcal{H}$, its generalization error with respect to $S$ is defined as: \begin{equation}\label{generr}\text{gen-err}_\mathcal{P}(\mathcal{A},S)=|L_{\mathcal{P}}(\mathcal{A}(S))-L_{S}(\mathcal{A}(S))|.\end{equation}
\end{dfn}
The definition above considers general loss functions. An important instance for the case of supervised learning is the $0-1$ loss. Suppose again that $\Z=\X\times\Y$ and that $\mathcal{H}=\{h|h:\X\to\Y\}$; given a pair $(x,y)\in\Z$ and a hypothesis $h:\X\to\Y$ the loss is defined as follows: \begin{equation}\ell(h,(x,y))=\mathbbm{1}_{h(x)\neq y}, \label{01loss}\end{equation} where $\mathbbm{1}$ is the indicator function. The corresponding errors become: \begin{equation}L_\mathcal{P}(h)= \E_{(x,y)\sim\mathcal{P}}[\mathbbm{1}_{h(x)\neq y}] = \mathcal{P}(\{(x,y): h(x)\neq y
\} )\end{equation} and \begin{equation}L_S(h)= \frac1n \sum_{i=1}^n \mathbbm{1}_{h(x_i)\neq y_i}.\label{empRisk}\end{equation}

Another fundamental concept we will need is the sample complexity of a learning algorithm. 
\begin{dfn}
	Fix $\epsilon, \delta \in (0,1)$. Let $\mathcal{H}$ be a hypothesis class. The sample complexity of $\mathcal{H}$ with respect to $(\epsilon,\delta)$, denoted by $m_\mathcal{H}(\epsilon,\delta)$,
	is defined as the smallest $m \in \mathbb{N}$ for which there exists a learning algorithm $\mathcal{A}$ such that, for every distribution $\mathcal{P}$ over the domain $\mathcal{X}$ we have that $\mathbb{P}(\text{gen-err}_\mathcal{P}(\mathcal{A},S)>\epsilon)\leq \delta.$
	If there is no such $m$ then $m_\mathcal{H}(\epsilon,\delta)=\infty$.
\end{dfn}
For more details we refer the reader to \cite{learningBook}.
\section{Main Results}
Our main theorem is a general bound on $\Pm_{XY} (E)$ in terms of $\Pm_X \Pm_Y(E)$, parameterized by two real numbers $\alpha$ and $\alpha'$. For particular choices of $\alpha$ and $\alpha'$, we demonstrate bounds in terms of $\alpha$-divergence, as well as $\alpha$-mutual information. The latter is a generalization of the maximal leakage bound in~\cite{ITW2019}.
\begin{theorem}\label{alphaExpBound}
	Let $(\X\times\Y,\F,\Pm_{XY}),(\X\times\Y,\F,\Pm_X\Pm_Y)$ be two probability spaces, and assume that $\Pm_{XY}\ll\Pm_X\Pm_Y$. Given $E\in\F$, let  $E_y := \{x : (x,y)\in E\}$, {\it i.e.,} the \enquote{fibers} of $E$ with respect to $y$. Then for any $E \in \F$,
	\begin{equation}
	\begin{split}
	\Pm_{XY}(E) 
	\leq &\left(\mathbb{E}_{\Pm_Y}\left[\Pm_X(E_Y)^{\gamma'/\gamma}\right]\right)^{1/\gamma'}\cdot\\&\left(\mathbb{E}_{\Pm_Y} \left[\E_{\Pm_X}^{\alpha'/\alpha}\left[ \left(\frac{dP_{XY}}{d\Pm_X\Pm_Y}\right)^\alpha \right]\right]\right)^{1/\alpha'}, \label{genBoundAlpha}
	\end{split}
	\end{equation}
	where $\gamma,\alpha,\gamma',\alpha'$ are such that $1=\frac1\alpha+\frac1\gamma = \frac{1}{\alpha'}+\frac{1}{\gamma'}$, and $\alpha,\gamma,\alpha',\gamma'\geq 1$.
\end{theorem} 
\begin{proof}
	We have that:
	\begin{align}
	\Pm_{XY}(E)&= \E_{\Pm_{XY}}[\mathbbm{1}_E]\\
	&= \E_{\Pm_X\Pm_Y}\left[\mathbbm{1}_E \frac{dP_{XY}}{d\Pm_X\Pm_Y}\right] \\
	&= \E_{\Pm_Y}\left[\E_{\Pm_X} \left[\mathbbm{1}_{\{X\in E_Y\}}\frac{dP_{XY}}{d\Pm_X\Pm_Y}\right]  \right] \\ \begin{split}
	&\leq \E_{\Pm_Y}  \!\!  \bigg[\left(\E_{\Pm_X} \left[\mathbbm{1}_{\{X\in E_Y\}}^\gamma\right]\right)^{1/\gamma} \cdot \\ & \qquad \left(\E_{\Pm_X} \left[\left(\frac{dP_{XY}}{d\Pm_X\Pm_Y}\right)^\alpha\right]\right)^{1/\alpha}  \bigg]\label{holder} \end{split} \\
	&=\E_{\Pm_Y}\left[\Pm_X(E_Y)^{1/\gamma}\left(\E_{\Pm_X}  \left[\left(\frac{dP_{XY}}{d\Pm_X\Pm_Y}\right)^\alpha\right]\right)^{1/\alpha} \right] \\
	\begin{split}&\leq
	\left(\mathbb{E}_{\Pm_Y}\left[\Pm_X(E_Y)^{\gamma'/\gamma}\right]\right)^{1/\gamma'}\cdot \\ 
	& \qquad \left(\mathbb{E}_{\Pm_Y}\left[\E_{\Pm_X}^{\alpha'/\alpha}\left[ \left(\frac{dP_{XY}}{d\Pm_X\Pm_Y}\right)^\alpha \right]\right]\right)^{1/\alpha'}, \label{holder2}
	\end{split}
	\end{align}
	where~\eqref{holder} and \eqref{holder2} follow from Holder's inequality, given that $\gamma,\alpha, \gamma',\alpha'\geq 1$ and $\frac1\gamma + \frac1\alpha = \frac1{\gamma'} + \frac1{\alpha'}=1$.
\end{proof}

\begin{remark} It is clear from the proof that one can similarly bound $\mathbb{E}[g(X,Y)]$ for any positive function $g(X,Y)$ such that $g(X,Y)$ is $\Pm_X\Pm_Y$-integrable. But the shape of the bound becomes more complex as one in general does not have that $g(X,Y)^\gamma=g(X,Y)$ for every $\gamma\geq 1$. 
\end{remark}
Based on the choices of $\alpha,\alpha'$, one can derive different bounds. Two are of particular interests to us and rely on different choices of $\alpha'$. Choosing $\alpha'=\alpha$ and thus $\gamma'=\gamma$ in Theorem 1, we retrieve:
\begin{corollary}\label{alphaDivBound}
	Let $(\X\times\Y,\F,\Pm_{XY}),(\X\times\Y,\F,\Pm_X\Pm_Y)$ be two probability spaces, and assume that $\Pm_{XY}\ll\Pm_X\Pm_Y$. Let $E\in\F$ we have that:
	\begin{align}
	P_{XY}(E)\leq &(\Pm_X\Pm_Y(E))^{\frac{\alpha-1}{\alpha}}\cdot\notag\\&\exp\left(\frac{\alpha-1}{\alpha}D_\alpha(\Pm_{XY}\|\Pm_X\Pm_Y)\right). 
	\end{align}
\end{corollary}
Choosing $\alpha'\to1$, which implies $\gamma'\to+\infty$, we retrieve:
\begin{corollary}\label{sibsMIBoundCor}
	Let $(\X\times\Y,\F,\Pm_{XY}),(\X\times\Y,\F,\Pm_X\Pm_Y)$ be two probability spaces, and assume that $\Pm_{XY}\ll\Pm_X\Pm_Y$. Given $E\in\F$, we have that:
	\begin{align}
	P_{XY}(E)\leq \!\! &\left(\! \esssup_{\Pm_Y} \Pm_X(E_Y)\right)^{1/\gamma}\cdot \\& \mathbb{E}_{\Pm_Y} \!\! \left[\E^{1/\alpha}_{\Pm_X}\left[ \left(\frac{dP_{XY}}{d\Pm_Y\Pm_X}\right)^\alpha\right]\right] \label{sibsNonVerdu}\\ = &\left(\esssup_{\Pm_Y} \Pm_X(E_Y)\right)^{\frac{\alpha-1}{\alpha}} \exp\left(\frac{\alpha-1}{\alpha} I_{\alpha}(X;Y)\right), \label{sibMIBound}
	\end{align}
	where $I_\alpha(X;Y)$ is the Sibson's mutual information of order $\alpha$ 
	\cite{verduAlpha}.
\end{corollary}
\begin{remark}
	An in-depth study of $\alpha-$mutual information appears in~\cite{verduAlpha}, where a slightly different notation is used. For reference, we can restate Eq. \eqref{sibsNonVerdu} in the notation of \cite{verduAlpha} to obtain:
	\begin{equation}
	\begin{split}
	P_{XY}(E)\leq &\left(\esssup_{\Pm_Y} \Pm_X(E_Y)\right)^{1/\gamma}\cdot \\ &\mathbb{E}_{\Pm_Y}\left[\E^{1/\alpha}_{\Pm_X}\left[ \left(\frac{dP_{Y|X}}{d\Pm_Y}\right)^\alpha \bigg| Y \right]\right].
	\end{split}
	\end{equation}
\end{remark}
Moreover, for a fixed $\alpha$ due to the property that Holder's conjugates need to satisfy, we have that $\frac1\gamma=\frac{\alpha-1}{\alpha}$ and the bound in \eqref{sibMIBound} can also be rewritten as:
\begin{equation}
\Pm_{XY}(E)\leq \exp\left(\frac{\alpha-1}{\alpha}\left( I_{\alpha}(X;Y)+\log\esssup_{\Pm_Y} \Pm_X(E_Y)\right)\right) \label{genBoundSibs}. 
\end{equation}
Considering the right hand side of \eqref{genBoundSibs}, because of the non-decreasability of Sibson's $\alpha-$Mutual Information with respect to $\alpha$ \cite{verduAlpha} we have that, for $1\le \alpha_1\leq \alpha_2$:  \begin{equation}\frac{\alpha_1-1}{\alpha_1}I_{\alpha_1}(X;Y)\leq \frac{\alpha_2-1}{\alpha_2}I_{{\alpha}_2}(X;Y).\end{equation}  Thus, choosing a smaller $\alpha$ yields a better dependence on $I_\alpha(X;Y)$ in the bound, but given that $\frac1\gamma= \frac{\alpha-1}{\alpha}$ we also have that $\frac{1}{\gamma_1}\leq \frac{1}{\gamma_2}$ and being $\esssup_{\Pm_Y} \Pm_X(E_Y)\leq 1$ it implies that \begin{equation}\left(\esssup_{\Pm_Y} \Pm_X(E_Y)\right)^\frac{1}{\gamma_1}\geq\left(\esssup_{\Pm_Y} \Pm_X(E_Y)\right)^\frac{1}{\gamma_2}, \end{equation} with a worse dependence on $\left(\esssup_{\Pm_Y}\Pm(E_Y)\right)^\frac{1}{\gamma}$ on the bound. This leads to a trade-off between the two quantities. If we focus on Corollary \ref{sibsMIBoundCor}, letting $\alpha\to\infty$ we recover a result involving maximal leakage \cite{ISIT2019,ITW2019}, but extending it to general alphabets:
\begin{corollary}\label{adaptML}
	Let $(\X\times\Y,\F,\Pm_{XY}),(\X\times\Y,\F,\Pm_X\Pm_Y)$ be two probability spaces, and assume that $\Pm_{XY}\ll\Pm_X\Pm_Y$. Let $E\in\F$ we have that:
	\begin{equation}
	P_{XY}(E) \leq \left(\esssup_{\Pm_Y} \Pm_X(E_Y)\right) \exp\left(\ml{X}{Y}\right), \label{MLBound}
	\end{equation}
	where $\ml{X}{Y}$ is the maximal leakage~\cite{leakage}.
\end{corollary}
The bound follows from the fact that $ \ml{X}{Y} = I_\infty(X;Y)$\cite{leakageLong}. 
A comparison between the bound for maximal leakage and some analogous result obtained for mutual information (through a different approach \cite{learningMI,infoThGenAn}) can be found in \cite{ITW2019}.
\section{Applications}
In this section, we consider some applications of the above bounds in the context of the generalization error.
In the bounds of interest $\Pm_X(E_y)$ is typically exponentially decaying with the number of samples and the trade-off between $\alpha$ and $\gamma$ can be explicitly seen in the sample complexity of a learning algorithm:
\begin{corollary}\label{generrSibs}
	Let $\X \times \Y$ be the sample space and $\mathcal{H}$ be the set of hypotheses.
	Let $\mathcal{A}:\mathcal{X}^n \times \mathcal{Y}^n\to \mathcal{H}$ be a learning algorithm that, given a sequence $S$ of $n$ points, returns a hypothesis $h\in \mathcal{H}$. Suppose $S$ is sampled i.i.d according to some distribution $\mathcal{P}$ over $\X \times \Y$, {\it i.e.,} $S\sim \mathcal{P}^n$. Let $\ell$ be the $0-1$ loss function as defined in \eqref{01loss}. Given $\eta \in (0,1)$, let $E=\{(S,h):|L_{\mathcal{P}}(h)-L_S(h)|>\eta \}$. Fix $\alpha\geq 1$. Then,
	\begin{equation}
	\mathbb{P}(E)\leq \exp\left(\frac{\alpha-1}{\alpha}\left(I_\alpha(S;\A(S))+\log2-2n\eta^2\right)\right).
	\end{equation}
\end{corollary}
\begin{proof}
	Fix $\eta\in (0,1)$ and $\alpha\geq 1$. Let $\frac{1}{\gamma}= \frac{\alpha-1}{\alpha}$. Let us denote with $E_h$ the fiber of $E$ over $h$ for some $h\in\mathcal{H}$, {\it i.e.,} $E_h=\{S : |L_\Pm(h)-L_S(h)|>\eta\}$. Consider $S,\hat{S}\in\{\X\times\Y\}^n$, where $S=((x_1,y_1),\ldots,(x_n,y_n))$ and $\hat{S}=((\hat{x}_1,\hat{y}_1),\ldots,(\hat{x}_n,\hat{y}_n))$. If $S,\hat{S}$ differ only in one position $j$, {\it i.e.,} $(x_i,y_i)=(\hat{x}_i,\hat{y}_i)\,\forall i\in[n]\setminus\{j\}$ and $(x_j,y_j)\neq (\hat{x}_j,\hat{y}_j)$ we have that for every $h\in\mathcal{H}$,
	\begin{equation}
	|L_S(h)-L_{\hat{S}}(h)|
	\leq \frac{1}{n}. \label{sensitLoss}
	\end{equation}
	By McDiarmid's inequality \cite{BLM2013Concentration}[Sec. 1.1] and Ineq. \eqref{sensitLoss} we have that for every hypothesis $h\in\mathcal{H}$, \begin{equation} \mathcal{P}_S(E_h) \leq 2\cdot \exp(-2n\eta^2).\label{mcDiarmids1} \end{equation} Then it follows from Corollary~\ref{sibsMIBoundCor} and Ineq. \eqref{mcDiarmids1} that:
	\begin{align}
	\mathbb{P}(E) &\leq \exp\left(\frac{\alpha-1}{\alpha}I_\alpha(S; \mathcal{A}(S))\right) (2\exp(-2n\eta^2))^\frac{\alpha-1}{\alpha}.
	\end{align}
\end{proof}
\begin{corollary}
	Let $\X \times \Y$ be the sample space and $\mathcal{H}$ be the set of hypotheses.
	Let $\mathcal{A}:\mathcal{X}^n \times \mathcal{Y}^n\to \mathcal{H}$ be a learning algorithm that, given a sequence $S$ of $n$ points, returns a hypothesis $h\in \mathcal{H}$. Suppose $S$ is sampled i.i.d according to some distribution $\mathcal{P}$ over $\X \times \Y$, {\it i.e.,} $S\sim \mathcal{P}^n$. Let $\ell$ be the $0-1$ loss function.  Given $\eta \in (0,1)$, let $E=\{(S,h):|L_{\mathcal{P}}(h)-L_S(h)|>\eta \}$. Fix $\alpha\geq 1$ then, in order to ensure a confidence of $\delta\in(0,1)$, {\it i.e.,} $\mathbb{P}(E)\leq \delta$, we need a number of samples $m$ satisfying:
	\begin{equation}\label{sampleComplexitySibsMI}
	m\geq \frac{I_\alpha(S;\A(S))+\log2+\gamma\log\left(\frac{1}{\delta}\right)}{2\eta^2}.
	\end{equation}
\end{corollary}
\begin{proof}
	From Corollary \ref{generrSibs} we have that $$\mathbb{P}(E)\leq \exp\left(\frac{\alpha-1}{\alpha}\left(I_\alpha(S;\A(S))+\log2-2n\eta^2\right)\right).$$ Fix $\delta \in (0,1)$, our aim is to have that:
	\begin{equation}
	\exp\left(\frac{\alpha-1}{\alpha}\left(I_\alpha(S;\A(S))+\log2-2n\eta^2\right)\right)\leq \delta,
	\end{equation}
	solving the inequality wrt $n$ gives us Equation \eqref{sampleComplexitySibsMI}. 
\end{proof}
Smaller $\alpha$ means that $I_\alpha(S;\A(S))$ will be smaller, but it will imply a larger value for $\gamma =\frac{\alpha}{\alpha-1}$ and thus a worse dependency on $\log(1/\delta)$ in the sample complexity.
Let $\Z$ be the sample space and $\mathcal{H}$ be the set of hypotheses. An immediate generalization of Corollary \ref{generrSibs} follows by considering loss functions such that for every fixed $h\in\mathcal{H},$ the random variable $l(h,Z)$ (induced by $Z$) is $\sigma^2-$sub Gaussian\footnote{Given a random variable $X$ we say that it is $\sigma^2$-sub-Gaussian if for every $\lambda\in\mathbb{R}$:
	$\mathbb{E}[e^{\lambda X}] \leq e^{\frac{\lambda^2\sigma^2}{2}}$.} for some $\sigma>0$.
\begin{corollary}\label{generrSibs2}
	Let $\mathcal{A}:\Z^n \to \mathcal{H}$ be a learning algorithm that, given a sequence $S$ of $n$ points, returns a hypothesis $h\in \mathcal{H}$. Suppose $S$ is sampled i.i.d according to some distribution $\mathcal{P}$ over $\Z$. Let $\ell:\mathcal{H}\times\Z 
	\to \mathbb{R}$ be a loss function such that $\ell(h,Z)$ is $\sigma$-sub Gaussian random variable for every $h\in\mathcal{H}$.  
	Given $\eta \in (0,1)$, let $E=\{(S,h):|L_{\mathcal{P}}(h)-L_S(h)|>\eta \}$. Fix $\alpha\geq 1$.  Then,
	\begin{align}\mathbb{P}(E) \leq \exp\left(\frac{1}{\gamma}\left(I_\alpha(S;\A(S))+\log2 -n\frac{\eta^2}{2\sigma^2}\right)\right).\end{align}
\end{corollary}
\begin{proof}
	Fix $\eta\in(0,1)$. Let us denote with $E_h$ the fiber of $E$ over $h$ for some $h\in\mathcal{H}$, {\it i.e.,} $E_h=\{S : |L_\Pm(h)-L_S(h)|>\eta\}$.  By assumption we have that $l(h,Z)$ is $\sigma-$sub Gaussian for every $h$. We can thus use Hoeffding's inequality for every hypothesis $h\in\mathcal{H},$ and retrieve that for every $h\in\mathcal{H}:$ \begin{equation}
	\mathcal{P}_S(E_h) \leq 2\cdot \exp\left(-n\frac{\eta^2}{2\sigma^2}\right). \label{hoeffdings} \end{equation}
	Then it follows from Corollary~\ref{sibsMIBoundCor} and Ineq. \eqref{hoeffdings} that:
	\begin{align}
	\mathbb{P}(E) &\leq \exp\left(\frac{\alpha-1}{\alpha}I_\alpha(S; \mathcal{A}(S))\right) \left(2\exp\left(-n\frac{\eta^2}{2\sigma^2}\right)\right)^\frac{\alpha-1}{\alpha}.
	\end{align}
\end{proof}
One important characteristic of these bounds is that they involve information-measures satisfying the data processing inequality~\cite{verduAlpha}. This means that all these results about generalization are \textbf{robust} to post-processing, {\it i.e.,} if the outcome of any learning algorithm with bounded $I_\alpha$ is processed further, the value of the information measure cannot increase. 
Another desirable property that would render the usage of such measures appealing in Adaptive Data Analysis is the Adaptive Composition property \cite{genAdap}. Alas, the lack of a definition of conditional Sibson's MI does not allows us, for the moment, to fully address the issue and verify whether or not the measure composes adaptively (like Mutual Information and Maximal Leakage \cite{infoThGenAn,ITW2019}).
Moreover, a comparison between this and other well-known results in the literature can be found in Table \ref{comparison}.  One can immediately see that the Sibson's MI bound and, in particular, the Maximal Leakage one, are the ones that most resemble the VC-Dimension bound both in terms of excess probability decay and sample complexity. 
\begin{table*}
	\begin{center}
		\caption{Comparison between bounds}
		\label{comparison}
		\begin{tabular}{c c c c c} 
			\toprule
			& Robust & Adaptive & Bound & Sample Complexity \\
			\midrule
			$\beta-$Stability \cite{bousqet}
			& No & No & exp. decay in $n$ & $f(\beta,\eta)\times\log\left(\frac{2}{\delta}\right)$ \\
			\addlinespace
			$\epsilon$-DP \cite{genAdap}
			& Yes & Yes &$\frac{1}{4} \exp{\left(\frac{-n\eta^2}{12}\right)}$, $\epsilon\leq \eta/2$ & $ \frac{12\cdot\log(1/4\delta)}{\eta^2}$\\
			\addlinespace
			MI \cite{learningMI}
			& Yes & Yes &$(I(X;Y) +1)/(2n\eta^2 -1)$ & $I(X;Y)/\eta^2\delta$ \\
			\addlinespace
			Maximal Leakage \cite{ITW2019}
			& Yes & Yes & $2\cdot\exp(\mathcal{L}(X\to Y)-2n\eta^2)$ & $(\ml{X}{Y} + \log\left(\frac{2}{\delta}\right))/2\eta^2$ \\
			\addlinespace
			$\alpha$-Sibson's MI 
			& Yes & Unknown & $\exp(\frac{\alpha-1}{\alpha}(I_\alpha(S;\A(S))+\log2-2n\eta^2))$ & $(I_\alpha(X;Y)+\log2+\gamma\log\left(\frac{1}{\delta}\right))/2\eta^2$ \\
			\addlinespace
			VC-Dim. $K$ \cite{learningBook}
			& & & $2\cdot\exp(\log(K)-2n\eta^2) $& 	$(\log(K)+\log\left(\frac{2}{\delta}\right))/2\eta^2$ \\
			\bottomrule
		\end{tabular}
		
	\end{center}
\end{table*}
\section{Bounds on Expected Generalization Error}
So far, when analyzing the generalization error, we have only considered high probability bounds, what can these results tell us about the \textbf{expected} generalization error? In order to provide a meaningful bound, some assumptions on the quantity $\max_h\Pm_{S}(|L_S(h)-\mathbb{E}[L(h)]|>\eta)$ are needed (where $S$ is a random vector of length $n$, sampled in an iid fashion from some distribution $\mathcal{D})$. More precisely, we will assume this probability to be exponentially decreasing with the number of samples $n$, as it often happens in the literature \cite{bousqet,BLM2013Concentration}. The following result is inspired by \cite[p. 419]{learningBook} with a slightly different proof.
\begin{lemma}\label{lemmaBoundExp}
	Let $X$ be a random variable and let $\hat{x}\in\mathbb{R}$. Suppose that exist $a\geq 0$ and $b\geq e$ such that for every $\eta>0$ $\Pm_X(|X-\hat{x}|\geq \eta)\leq 2b\exp\left(-\eta^2/a^2\right)$ then $\mathbb{E}\left[|X-\hat{x}|\right]\leq a\left(\sqrt{\log 2b}+\frac{1}{2\sqrt{\log 2b}}\right)$.
\end{lemma}
\begin{proof}
	\begin{align}
	\mathbb{E}\left[|X-\hat{x}|\right]&=\int_0^{+\infty} \Pm_X(|X-\hat{x}|\geq \eta)d\eta \\
	&\leq \int_0^{+\infty} \min\left(1,2b\exp\left(-\eta^2/a^2\right)\right) d\eta \label{multFactor} \\ 
	&= \int_0^{\sqrt{a^2\log 2b}} d\eta + \int_{\sqrt{a^2\log 2b}}^{+\infty} 2b\exp(-\frac{\eta^2}{a^2})d\eta  \\
	&\leq a\left(\sqrt{\log 2b}+ \frac{1}{2\sqrt{\log 2b}}\right).
	\end{align}
\end{proof}
\begin{theorem}\label{boundExp}
	Let $\mathcal{A}:\mathcal{Z}^n\to \mathcal{H}$ be a learning algorithm and let $I_\alpha(S;\mathcal{A}(S))$ be the dependence measure chosen. Suppose that the loss function $l:\mathcal{Z}\times\mathcal{H}\to \mathbb{R}$ is such that $\forall h \Pm_{S\sim \mathcal{D}^n}(|L_S(h)-\mathbb{E}[L(h)]|>\eta)\leq 2\exp\left(-\frac{\eta^2}{2\sigma^2}n\right) $ for some $\sigma>0$ (e.g. $l(h,Z)$ is $\sigma^2$-sub-Gaussian), then:
	\begin{align}
	&\mathbb{E}\left[|L_S(H)-\mathbb{E}[L(H)]|\right] \leq\\ &\sqrt{\frac{2\sigma^2 \gamma}{n}}\left(\sqrt{\frac{\log(2)+I_\alpha(S;\mathcal{A}(S))}{\gamma}}+  \frac{1}{2\sqrt{\frac{\log2+I_\alpha(S;\mathcal{A}(S))}{\gamma}}}\right).
	\end{align}
\end{theorem}
\begin{proof}
	The proof is a simple application of Lemma \ref{lemmaBoundExp} and Corollary \ref{generrSibs2} with $a=\sqrt{2\gamma\sigma^2}/\sqrt{n}$ and with $b=2^{\frac1\gamma -1}\exp\left(\frac{I_\alpha(\mathcal{A}(S);S)}{\gamma}\right)$.
\end{proof}
An interesting application of Theorem \ref{boundExp} can be found by considering $\ml{S}{\A(S)}$ and the $0-1$ loss (hence, $1/4$-sub-Gaussian). 
\begin{corollary}
	Let $\mathcal{A}:\mathcal{Z}^n\to \mathcal{H}$. Consider the $0-1$ loss, then $\forall h \Pm_{S\sim \mathcal{D}^n}(|L_S(h)-\mathbb{E}[L(h)]|>\eta)\leq 2\exp\left(-2\eta^2n\right)$, and:
	\begin{align}
	&\mathbb{E}\left[|L_S(H)-\mathbb{E}[L(H)]|\right] \leq \\ &\frac{1}{\sqrt{2n}}\left(\sqrt{\log2+\ml{S}{\A(S)}}+ \frac{1}{2\sqrt{\log2+ \ml{S}{\A(S)}}}\right).
	\end{align}
\end{corollary}

\bibliographystyle{IEEEtran}
\bibliography{sample}

\begin{thebibliography}{10}
\providecommand{\url}[1]{#1}
\csname url@samestyle\endcsname
\providecommand{\newblock}{\relax}
\providecommand{\bibinfo}[2]{#2}
\providecommand{\BIBentrySTDinterwordspacing}{\spaceskip=0pt\relax}
\providecommand{\BIBentryALTinterwordstretchfactor}{4}
\providecommand{\BIBentryALTinterwordspacing}{\spaceskip=\fontdimen2\font plus
\BIBentryALTinterwordstretchfactor\fontdimen3\font minus
  \fontdimen4\font\relax}
\providecommand{\BIBforeignlanguage}[2]{{%
\expandafter\ifx\csname l@#1\endcsname\relax
\typeout{** WARNING: IEEEtran.bst: No hyphenation pattern has been}%
\typeout{** loaded for the language `#1'. Using the pattern for}%
\typeout{** the default language instead.}%
\else
\language=\csname l@#1\endcsname
\fi
#2}}
\providecommand{\BIBdecl}{\relax}
\BIBdecl

\bibitem{learningMI}
R.~Bassily, S.~Moran, I.~Nachum, J.~Shafer, and A.~Yehudayoff, ``Learners that
  use little information,'' ser. Proceedings of Machine Learning Research,
  vol.~83.\hskip 1em plus 0.5em minus 0.4em\relax PMLR, 07--09 Apr 2018, pp.
  25--55.

\bibitem{infoThGenAn}
A.~{Xu} and M.~{Raginsky}, ``{Information-theoretic analysis of generalization
  capability of learning algorithms},'' in \emph{Advances in Neural Information
  Processing Systems}, 2017, p. 2521–2530.

\bibitem{explBiasMI}
D.~Russo and J.~Zou, ``Controlling bias in adaptive data analysis using
  information theory,'' in \emph{Proceedings of the 19th International
  Conference on Artificial Intelligence and Statistics}, ser. Proceedings of
  Machine Learning Research, vol.~51.\hskip 1em plus 0.5em minus 0.4em\relax
  PMLR, 09--11 May 2016, pp. 1232--1240.

\bibitem{leakage}
I.~Issa, S.~Kamath, and A.~B. Wagner, ``An operational measure of information
  leakage,'' in \emph{2016 Annual Conference on Information Science and Systems
  (CISS)}, March 2016, pp. 234--239.

\bibitem{ISIT2019}
I.~Issa, A.~R. Esposito, and M.~Gastpar, ``Strengthened information-theoretic
  bounds on the generalization error,'' in \emph{2019 {IEEE} International
  Symposium on Information Theory, {ISIT} Paris, France, July 7-12}, 2019.

\bibitem{ITW2019}
A.~R. Esposito, M.~Gastpar, and I.~Issa, ``Learning and adaptive data analysis
  via maximal leakage,'' in \emph{{IEEE} Information Theory Workshop, {ITW}
  2019, Visby, Gotland, Sweden, Aug 25-28}, 2019.

\bibitem{verduAlpha}
S.~Verd{\'{u}}, ``{\(\alpha\)}-mutual information,'' in \emph{2015 Information
  Theory and Applications Workshop, {ITA} 2015, San Diego, CA, USA, February
  1-6, 2015}, 2015, pp. 1--6.

\bibitem{alphaDiv}
T.~{van Erven} and P.~{Harremos}, ``R\'enyi divergence and kullback-leibler
  divergence,'' \emph{IEEE Transactions on Information Theory}, vol.~60, no.~7,
  pp. 3797--3820, July 2014.

\bibitem{opMeanRDiv1}
I.~{Csiszar}, ``Generalized cutoff rates and r\'enyi's information measures,''
  \emph{IEEE Transactions on Information Theory}, vol.~41, no.~1, pp. 26--34,
  Jan 1995.

\bibitem{opMeanRDiv2}
P.~D. Gr\"{u}nwald, \emph{The Minimum Description Length Principle (Adaptive
  Computation and Machine Learning)}.\hskip 1em plus 0.5em minus 0.4em\relax
  The MIT Press, 2007.

\bibitem{learningBook}
S.~Shalev-Shwartz and S.~Ben-David., \emph{Understanding machine learning: From
  theory to algorithms}.\hskip 1em plus 0.5em minus 0.4em\relax Cambridge
  University Press, 2014.

\bibitem{leakageLong}
I.~{Issa}, A.~B. {Wagner}, and S.~{Kamath}, ``{An Operational Approach to
  Information Leakage},'' \emph{ArXiv e-prints}, jul 2018.

\bibitem{BLM2013Concentration}
S.~Boucheron, G.~Lugosi, and P.~Massart, \emph{Concentration Inequalities: A
  Nonasymptotic Theory of Independence}.\hskip 1em plus 0.5em minus 0.4em\relax
  Oxford University Press, 2013.

\bibitem{genAdap}
C.~Dwork, V.~Feldman, M.~Hardt, T.~Pitassi, O.~Reingold, and A.~Roth,
  ``Generalization in adaptive data analysis and holdout reuse,'' in
  \emph{Proceedings of the 28th International Conference on Neural Information
  Processing Systems - Volume 2}.\hskip 1em plus 0.5em minus 0.4em\relax
  Cambridge, MA, USA: MIT Pressf, 2015.

\bibitem{bousqet}
O.~Bousquet and A.~Elisseeff, ``Stability and generalization,'' \emph{J. Mach.
  Learn. Res.}, vol.~2, pp. 499--526, 3 2002.

\end{thebibliography}

\end{document}